\documentclass[aip,jmp]{revtex4-1}

\usepackage{bbm,mathrsfs}
\usepackage{graphicx}
\usepackage{amsfonts, amsmath, amssymb}
\usepackage{amsthm}
\usepackage{psfrag}

\def\textbf#1{{\bf #1}}
\def\be{\begin{equation}}
\def\ee{\end{equation}}
\def\ben{\begin{eqnarray}}
\def\een{\end{eqnarray}}
\def\eea{\end{array}}
\def\bea{\begin{array}}

\newcommand{\bei}{\begin{itemize}}
\newcommand{\eei}{\end{itemize}}

\newcommand{\ii}[0]{\mathrm{i}}

\newtheorem{thm}{Theorem}
\newtheorem{cor}{Corollary}

\theoremstyle{definition}

\begin{document}

\title{Convergence of expansions in Schr\"odinger and Dirac eigenfunctions, with an application to the $R$-matrix theory}

\author{Julia Stasi\'{n}ska}
\email{julsta@ifae.es}

\affiliation{Atomic Physics Division, Department of Atomic Physics and Luminescence,\\
Faculty of Applied Physics and Mathematics, Gda\'nsk University of Technology,\\
Narutowicza 11/12, PL 80--233 Gda\'{n}sk, Poland} \affiliation{Grup de
F\'{i}sica Te\`{o}rica: Informaci\'{o} i Fen\'{o}mens Qu\`{a}ntics, Departament
de F\'{i}sica, Universitat Aut\`{o}noma de Barcelona, 08193 Bellaterra
(Barcelona), Spain}

\date{03.08.2011}

\begin{abstract}
Expansion of a wave function in a basis of eigenfunctions of a differential
eigenvalue problem lies at the heart of the $R$-matrix methods for both the
Schr\"{o}dinger and Dirac particles. A central issue that should be carefully
analyzed when functional series are applied is their convergence. In the
present paper, we study the properties of the eigenfunction expansions
appearing in nonrelativistic and relativistic $R$-matrix theories. In
particular, we confirm the findings of Rosenthal [J. Phys. G {\bf 13}, 491
(1987)] and Szmytkowski and Hinze [J. Phys. B {\bf 29}, 761 (1996); J. Phys. A
{\bf 29}, 6125 (1996)] that in the most popular formulation of the $R$-matrix
theory for Dirac particles, the functional series fails to converge to a limit
claimed by other authors.
\end{abstract}


\pagestyle{plain}

\maketitle
\section{Introduction}\label{intro}
Convergence of expansions of a one-component function into a series of
eigenfunctions of a Sturm--Liouville problem was a subject of many studies. In
some physical situations \cite{WignerEisenbud,SturmLiouville_applications}, of
particular interest are expansions of a function defined only on a finite and
closed interval. The classical results on convergence of such series can be
found, for instance, in
\onlinecite{SturmLiouville_applications,Titchmarsh,LevitanSargsjan,CoddingtonLevinson,Atkinson}
and for modern studies on this type of problems, including the equiconvergence
method, the reader is referred to \onlinecite{Minkin, Mityagin_Schrodinger} and
references therein.

A similar problem for a two-component function was addressed by a number of
mathematicians at the beginning of the 20th century
\cite{Hurwitz,Camp,Schur,BirkhoffLanger,Bliss,Titchmarsh_art} and later
reviewed in numerous textbooks (see e.g.
\onlinecite{LevitanSargsjan,CoddingtonLevinson,Atkinson}). However, very few
articles and textbooks deal with the development of an arbitrary function on a
closed interval. Usually, either the expansion on an open interval is studied
only \cite{Camp,Titchmarsh_art} or some additional conditions are imposed on
the expanded function at the boundary points\cite{Hurwitz, Camp, Schur, Bliss}.
To the best of the author's knowledge, the only classical paper discussing the
general situation is the one by Birkhoff and Langer \cite{BirkhoffLanger}. A
more recent analysis of this kind of problems can be found, e.g., in
\onlinecite{Mityagin_Dirac}. A generalization of the equiconvergence method
\cite{Minkin} to a vector case should be also possible.

In the present paper, we apply the general results concerning convergence of
eigenfunction expansions in the context of the $R$-matrix theory of scattering
processes.
This theory was first developed for low-energy collisions that could be
described with the Schr\"{o}dinger equation \cite{WignerEisenbud} (see
\onlinecite{ThomasLane,Barrett,Szmytkowski_rev} for reviews on the subject).
The $R$-matrix theory for the Dirac equation \cite{Goertzel} was formulated
soon after the nonrelativistic one, with nuclear applications in view. Only
later it was realized that the electron--atom collisions involving targets with
large atomic numbers require a Dirac description, due to the increasing role of
relativistic effects. The $R$-matrix theory was reinvestigated in this context
in \onlinecite{Chang}.

The central idea in the formulation of the $R$-matrix methods for scattering
from spherically symmetric potentials is to divide the whole space into two
regions, a finite reaction volume $r< \varrho$ and the outer region
$r\geq\varrho$, and to expand a wave function in the inner region in a series
of eigenfunctions of the Hamiltonian governing the scattering process
augmented, however, by \emph{artificial boundary conditions} at the sphere
$r=\varrho$.
This procedure allows one to express the $R$-matrix as a limit $r\to
\varrho^{-}$ of an infinite functional series. In both nonrelativistic and
relativistic theories the critical issue, the convergence of the series on the
boundary, was not properly analyzed by the originators and only presumed to
hold. The convergence question was first recognized by Rosenthal
\cite{Rosenthal}, however his conclusions were incorrect. Later Szmytkowski and
Hinze \cite{SzmytkowskiHinze1995,SzmytkowskiHinze1996,Szmytkowski1998} realized
that while the development of the solution in the Schr\"odinger formulation
converges in the whole interval to an \emph{expanded} function, the analogous
series appearing in the relativistic case has a discontinuity at the crucial
boundary point. In this way, in the most popular formulation of the method, the
solution depends on the \emph{artificial} boundary condition imposed on the
basis functions. Taking this into account, Szmytkowski and Hinze developed the
correct Dirac $R$-matrix theory. Their conclusion caused much controversy
\cite{Grant_last} and was not widely recognized by the community
\cite{Grant_book, Burke_book}. We present here a theorem confirming their
results\cite{SzmytkowskiHinze1995,SzmytkowskiHinze1996,Szmytkowski1998} as well
as the general result on convergence obtained by Szmytkowski \cite{JMathPhys}.

The paper is organized as follows. In Section \ref{expansions}, we recall basic
facts from both nonrelativistic and relativistic $R$-matrix theories to
highlight the problem of convergence appearing in both of them. In Section
\ref{convergence}, we give the general convergence theorems concerning the
eigenfunction expansions \cite{Titchmarsh,BirkhoffLanger}. The main result of
the paper, a solution to the Dirac $R$-matrix puzzle based on the theorem by
Birkhoff and Langer \cite{BirkhoffLanger}, can be found in Section
\ref{convergence}. We finish the paper with conclusions and point out some open
problems.


\section{Expansions appearing in the $R$-matrix theories}\label{expansions}
The nonrelativistic and relativistic theories share many similarities, however
in one essential point they are very different, i.e. the eigenfunction
expansion of the solution of a nonrelativistic wave equation converges to a
continuous function, whereas an analogous series in the relativistic theory has
a discontinuity at the crucial boundary point
\cite{SzmytkowskiHinze1995,SzmytkowskiHinze1996,JMathPhys}. As a result, the
relativistic $R$-matrix is not appropriately expressed by a functional series.
To highlight this difference, we shortly introduce both methods in a
single-channel scattering from spherically symmetric potentials. The notation
used in the following sections is based on the monograph on the $R$-matrix
methods in scattering \cite{Szmytkowski_rev}.
%
\subsection{Nonrelativistic $R$-matrix theory}
A nonrelativistic elastic scattering process of spinless particles with mass $m$ and energy $E>0$ from a spherically symmetric potential $V(r)$ is governed by the stationary Schr\"{o}dinger equation.
We assume that the potential $V(r)$ affects the particle only in a finite
spherical volume of radius $\varrho$ centered at $r=0$, denoted further by
$\mathcal{V}_{\varrho}$. Outside this (inner) region the particle is free and
its wave function satisfies the free-Hamiltonian stationary Schr\"{o}dinger
equation. Obviously, the solution in the inner region must pass smoothly into
the solution in the outer region. We denote by $\Psi(E,\mathbf{r})$ the wave
function being the solution of the respective Schr\"{o}dinger equations in the
inner and outer regions.

Since the potential is spherically symmetric, it is enough to consider only the radial part of the function $\Psi(E,\mathbf{r})$ corresponding to the multiindex $\gamma=(l,m_l)$. For a general function $f(\mathbf{r})$, it is defined as
\begin{equation}\label{radial_schr}
F_{\gamma}(r)=\int_{4\pi}\mathrm{d}\mathbf{\hat{r}} r^2 \Upsilon_{\gamma}(\mathbf{r}) f(\mathbf{r}),\quad \mathbf{\hat{r}}=\frac{\mathbf{r}}{r},
\end{equation}
where $\Upsilon_\gamma(\mathbf{r})=(1/r) \mathrm{i}^{l}
Y_{\gamma}(\mathbf{\hat{r}})$, and $Y_{\gamma}$ are normalized spherical
harmonics defined as in \onlinecite{CondonShortley}. We denote by
$\boldsymbol{P}(E,r)$ a vector of radial functions of $\Psi(E,\mathbf{r})$ with elements $P_{\gamma}(E,r)$, and by
$\boldsymbol{D}(E,r)$ a vector of radial functions of $\mathbf{\hat{r}}\cdot\boldsymbol{\nabla}\Psi(E,\mathbf{r})$ with elements $D_{\gamma}(E,r)$.
Let us assume that there exists a matrix $\mathsf{R}_{\mathsf{b}}(E,\varrho)$
connecting $\boldsymbol{P}(E,r)$ and $\boldsymbol{D}(E,r)$ on the boundary of
$\mathcal{V}_{\varrho}$ (which on the radial grid corresponds to $r=\varrho$)
in the following way:
\begin{equation}\label{rmatrix1}
\boldsymbol{P}(E,\varrho)=\mathsf{R}_{\mathsf{b}}(E,\varrho)[\boldsymbol{D}(E,\varrho)-\mathsf{b}\boldsymbol{P}(E,\varrho)],
\end{equation}
where $\mathsf{b}$ is an arbitrary square matrix. The above relation defines
the $R$-matrix $\mathsf{R}_{\mathsf{b}}(E,\varrho)$. In what follows, it will
be assumed that $\mathsf{b}$ is a diagonal, energy-independent, real matrix.
Then the $R$-matrix is also diagonal, and its elements will be denoted by
$(\mathsf{R}_{\mathsf{b}})_{\gamma\gamma}\equiv\mathsf{R}_{\mathsf{b}\gamma}$.
Finding the $R$-matrix is equivalent to solving the scattering problem since
$\mathsf{R}_{\mathsf{b}}(E,\varrho)$ is simply connected to the scattering
matrix \cite{Szmytkowski_rev}.

To determine the eigenfunction expansion of the $R$-matrix, we consider the radial part of the Schr\"{o}dinger equation in the inner region $\mathcal{V}_{\varrho}$:
\begin{eqnarray}\label{schr_rad}
\left(-\frac{\hslash^2}{2m}\frac{\mathrm{d}^2}{\mathrm{d}r^2}+\frac{\hslash^2 l(l+1)}{2mr^2}+V(r)-E\right)P_{\gamma}(E,r)=0,\quad r\in[0,\varrho).&&
\end{eqnarray}
We emphasize that here we do \emph{not} make any restrictions on
the function, except that it vanishes at $r=0$ as $r^{l+1}$. Our
aim is to expand the unknown radial function $P_{\gamma}(E,r)$ in
the basis $\{P_i^{(\gamma)}(r)\}$ generated by the same
Hamiltonian, but augmented by the \emph{artificial} boundary
condition at $r=\varrho$, that is
\begin{eqnarray}
\left(-\frac{\hslash^2}{2m}\frac{\mathrm{d}^2}{\mathrm{d}r^2}+\frac{\hslash^2l(l+1)}{2mr^2}+V(r)-E_i\right)P_i^{(\gamma)}(r)=0, \quad r\in[0,\varrho],&&\label{schr_eig}\\
\lim_{r\to 0} r^{-l-1} P_i^{(\gamma)}(r)=\mathrm{const}, \qquad
\frac{\mathrm{d}}{\mathrm{d}r}P_i^{(\gamma)}(r)\Big|_{r=\varrho}=\left(\mathsf{b}_{\gamma\gamma}+\frac{1}{\varrho}\right)P_i^{(\gamma)}(\varrho).\label{schr_eig_bound}
\end{eqnarray}
The set of eigenvalues $E_i$ is countably infinite, and eigenfunctions
corresponding to different eigenvalues are orthogonal under the standard scalar
product in $\mathcal{L}^2([0,\varrho])$. The formal expansion of an arbitrary
function on $[0,\varrho]$ is
\begin{eqnarray}\label{exp_schr}
P_{\gamma}(E,r)=\sum_{i=0}^{\infty} C_i(E) P_i^{(\gamma)}(r),\qquad C_i(E)=\int_{0}^{\rho}\mathrm{d}rP_i^{(\gamma)}(r)P_{\gamma}(E,r),&&
\end{eqnarray}
where we assume that the functions $\{P_i^{(\gamma)}\}$ are normalized to
unity. The choice of the boundary conditions (\ref{schr_eig_bound}) for the
basis functions allows us to write the coefficient $C_i(E)$ in a form that
reveals the proportionality to
$D_{\gamma}(E,\varrho)-\mathsf{b}_{\gamma\gamma}P_{\gamma}(E,\varrho)$ [compare
to equation (\ref{rmatrix1})]. In the case that we consider it holds that
$D_{\gamma}(E,r)=r\partial_r (1/r) P_{\gamma}(E,r)$. Using equations
(\ref{schr_rad}) and (\ref{schr_eig}) and the boundary conditions fulfilled by
$P_i^{(\gamma)}(r)$, we obtain
\begin{equation}
P_{\gamma}(E,r)=\frac{\hslash^2}{2m}[D_{\gamma}(E,\varrho)-\mathsf{b}_{\gamma\gamma}P_{\gamma}(E,\varrho)]\sum_{i=0}^{\infty}
\frac{P_i^{(\gamma)}(\varrho)}{E_i-E}P_i^{(\gamma)}(r),\quad r\in[0,\varrho).
\end{equation}
Taking the limit $r\to\varrho^{-}$ on both sides leads to
\begin{equation}\label{rmatrix2}
P_{\gamma}(E,\varrho)=\mathsf{R}_{\mathsf{b}\gamma}(E,\varrho)[D_{\gamma}(E,\varrho)-\mathsf{b}_{\gamma\gamma}P_{\gamma}(E,\varrho)],
\end{equation}
where
\begin{equation}\label{rmatrix_el}
\mathsf{R}_{\mathsf{b}\gamma}(E,\varrho)=\frac{\hslash^2}{2m} \lim_{r\to\varrho^{-}}\sum_{i=0}^{\infty} \frac{P_i^{(\gamma)}(\varrho)P_i^{(\gamma)}(r)}{E_i-E}.
\end{equation}
Comparing equations (\ref{rmatrix2}) and (\ref{rmatrix1}) we see that equation
(\ref{rmatrix_el}) defines a diagonal element of the $R$-matrix
$\mathsf{R}_{\mathsf{b}}$. One notices that the $R$-matrix is expressed as a
continuous extension of a functional series to the point $r=\varrho$. The main
question that we are going to answer in section 3 is: {\it Can one interchange
the symbols of limit and sum, and still obtain the same result?} In other
words, does the series on the right-hand side of equation (\ref{rmatrix_el})
converge to a continuous function in $[0,\varrho]$?
%
\subsection{Relativistic $R$-matrix theory}
The relativistic description of an elastic scattering process for particles of
spin $\frac{1}{2}$ with rest mass $m$ and total energy $E$ ($|E|>mc^2$) is
governed by the stationary Dirac equation.
Similarly as we have done in the nonrelativistic case, we assume that the
potential $V(r)$ vanishes outside the spherical volume $\mathcal{V}_{\varrho}$
bounded by a spherical shell, corresponding on the radial grid to $r=\varrho$,
and that the wave functions $\boldsymbol{\Psi}(E,\mathbf{r})$ in the inner and
outer regions pass smoothly one into the other on the boundary of
$\mathcal{V}_\varrho$.
To define the relativistic $R$-matrix, we fix the following notation:
\begin{equation}
\Omega_{\gamma}^{(+)}(\mathbf{r})=\frac{1}{r}
\left(\!\!
\begin{array}{c}
  \mathrm{i}^l \Omega_{\kappa m_j}(\mathbf{\hat{r}}) \\
  0
\end{array}\!\!
\right),\quad
\Omega_{\gamma}^{(-)}(\mathbf{r})=\frac{1}{r}
\left(\!\!
\begin{array}{c}
  0 \\
  \mathrm{i}^{l+1} \Omega_{-\kappa m_j}(\mathbf{\hat{r}})
\end{array}\!\!
\right),
\end{equation}
where the multiindex $\gamma$ is defined as $(\kappa,m_j)$, with $\kappa\in
\mathbbm{Z}\setminus \{0\}$ and
$m_j=\{-|\kappa|+1/2,-|\kappa|+3/2,\ldots,|\kappa|-1/2\}$,
$l=|\kappa+1/2|-1/2$, and $\Omega_{\pm\kappa m_j}(\hat{\mathbf{r}})$ are the
spherical spinors \cite{Szmytkowski_spinors}. We define two radial functions of
a four-component vector $\mathbf{f}(\mathbf{r})$, denoted by the superscripts
$\pm$. For a fixed multiindex $\gamma$ they are given by
\begin{equation}
F_{\gamma}^{(\pm)}(r)=\int_{4\pi}\mathrm{d}\mathbf{\hat{r}} r^2 \Omega_{\gamma}^{(\pm)\dagger}(\mathbf{r}) \mathbf{f}(\mathbf{r}).
\end{equation}
We denote by
$\boldsymbol{P}(E,r)$ and $\boldsymbol{Q}(E,r)$ vectors of "$+$" and "$-$"
radial functions of $\boldsymbol{\Psi}(E,\mathbf{r})$, respectively, with
elements $P_{\gamma}(E,r)$ and $Q_{\gamma}(E,r)$.
Let us define the $R$-matrix $\mathsf{R}_{\mathsf{b}}^{(+)}(E,\varrho)$
connecting $\boldsymbol{P}(E,r)$ and $\boldsymbol{Q}(E,r)$ on the surface of
$\mathcal{V}_{\varrho}$ in the following way:
\begin{eqnarray}
\boldsymbol{P}(E,\varrho)=\mathsf{R}_{\mathsf{b}}^{(+)}(E,\varrho)\left[\left(\frac{2mc}{\hslash}\right)\boldsymbol{Q}(E,\varrho)-\mathsf{b}\boldsymbol{P}(E,\varrho)\right],&&\label{rmatrix_dir1}
\end{eqnarray}
where $\mathsf{b}$ is some square matrix. Henceforward we will assume $\mathsf{b}$ to be diagonal, energy-independent and real. In this case the $R$-matrix will be diagonal as well.

To find the $R$-matrix, we consider the radial part of the Dirac equation in the internal region $\mathcal{V}_{\varrho}$:
\begin{eqnarray}\label{dir_rad}
\left(
  \begin{array}{cc}
    mc^2+V(r)-E & c\hslash (-\mathrm{d}/\mathrm{d}r+\kappa/r) \\
    c\hslash(\mathrm{d}/\mathrm{d}r+\kappa/r) & -mc^2+V(r)-E
  \end{array}\right)
\left(
  \begin{array}{c}
    P_{\gamma}(E,r) \\
    Q_{\gamma}(E,r)
  \end{array}
\right)
=0,\quad r\in[0,\varrho).
\end{eqnarray}
Since the solution must fulfill some unknown boundary condition,
given by (\ref{rmatrix_dir1}), at the point $r=\varrho$, we do not
make any assumptions on the functions $P_{\gamma}(E,r)$ and
$Q_{\gamma}(E,r)$, except for that they vanish for $r=0$. Note
that equation (\ref{dir_rad}) is a homogeneous differential
equation in which $E$ is a parameter, and not an eigenvalue
problem. We will expand the unknown solution in the basis
generated by the eigenproblem consisting of the Dirac Hamiltonian
of the previous equation and boundary conditions that, though
\emph{unphysical}, will allow us to develop the $R$-matrix into a
functional series. Let us consider the eigenproblem
\begin{eqnarray}
&&\left(\!\!
  \begin{array}{cc}
    mc^2+V(r)-E_i &\!\! c\hslash (-\mathrm{d}/\mathrm{d}r+\kappa/r) \\
    c\hslash(\mathrm{d}/\mathrm{d}r+\kappa/r) &\!\! -mc^2+V(r)-E_i
  \end{array}\!\!\right)
\left(\!\!
  \begin{array}{c}
    P_i^{(\gamma)}(r) \\
    Q_i^{(\gamma)}(r)
  \end{array}\!\!
\right)
=0,\quad r\in[0,\varrho],\label{dir_eig}\\
&&\lim_{r\to 0}r^{-\nu}P_i^{(\gamma)}(r)=\mathrm{const}, \qquad  Q_i^{(\gamma)}(\varrho)=(2mc/\hslash)^{-1}\mathsf{b}_{\gamma\gamma}P_i^{(\gamma)}(\varrho),\label{dir_eig_bc}
\end{eqnarray}
where $\nu=l+1$ if $V(0)=\mathrm{const}$, and $\nu=\sqrt{\kappa^2-(\alpha
Z)^2}$ for the Coulomb potential ($\alpha$ is the fine-structure constant and
$Z$ the atomic number).
The set of real eigenvalues $\{E_i\}$ is infinitely countable. Moreover, the
eigenfunctions corresponding to different eigenvalues are orthogonal in the
sense
\begin{equation}
\int_{0}^{\varrho}\mathrm{d}r \big(P_i^{(\gamma)}(r), Q_i^{(\gamma)}(r) \big)
\left(\begin{array}{c}
P_j^{(\gamma)}(r)\\
Q_j^{(\gamma)}(r)
\end{array}\right)=\mathcal{N}_i^2\delta_{ij}.
\end{equation}
We will further assume the eigenfunctions to be normalized to unity; then the
formal expansion of the functions $P_{\gamma}(E,r)$ and $Q_{\gamma}(E,r)$ is
given by
\begin{eqnarray}\label{exp_dir}
&&\left(
  \begin{array}{c}
    P_{\gamma}(E,r) \\
    Q_{\gamma}(E,r)
  \end{array}
\right)
=\sum_{i=-\infty}^{\infty} C_i(E)
\left(
  \begin{array}{c}
    P_i^{(\gamma)}(r) \\
    Q_i^{(\gamma)}(r)
  \end{array}
\right),\quad r\in[0,\varrho),\\[2ex]
&&C_i(E)=\int_{0}^{\varrho}\mathrm{d}r \Big(P_i^{(\gamma)}(r),Q_i^{(\gamma)}(r)\Big)
\left(
  \begin{array}{c}
    P_{\gamma}(E,r) \\
    Q_{\gamma}(E,r)
  \end{array}
\right).
\end{eqnarray}
The coefficients $C_i(E)$ can be written in the form revealing the connection
to the $R$-matrix. Using equations (\ref{dir_rad}) and (\ref{dir_eig}) and the
boundary conditions fulfilled by $P_i^{(\gamma)}(r)$, we obtain
\begin{eqnarray}\label{exp_dir2}
&&\left(\!\!
  \begin{array}{c}
    P_{\gamma}(E,r) \\
    Q_{\gamma}(E,r)
  \end{array}\!\!
\right)
=\frac{\hslash^2}{2m}\left[\frac{2mc}{\hslash}Q_{\gamma}(E,\varrho)-\mathsf{b}_{\gamma\gamma}P_{\gamma}(E,\varrho)\right]
\!\sum_{i=-\infty}^{\infty}\! \frac{P_i^{(\gamma)}(\varrho)}{E_i-E}
\left(\!\!
  \begin{array}{c}
    P_i^{(\gamma)}(r) \\
    Q_i^{(\gamma)}(r)
  \end{array}\!\!
\right),\quad r\in[0,\varrho).\nonumber
\end{eqnarray}
Taking the limit $r\to\varrho^{-}$ on both sides we obtain for the upper component
\begin{equation}\label{rmatrix2_dir}
    P_{\gamma}(E,\varrho)=
    \mathsf{R}_{\mathsf{b}\gamma}^{(+)}(E,\varrho)
    \left[\displaystyle{\frac{2mc}{\hslash}}Q_{\gamma}(E,\varrho)-\mathsf{b}_{\gamma\gamma}P_{\gamma}(E,\varrho) \right],
\end{equation}
where
\begin{eqnarray}
\mathsf{R}_{\mathsf{b}\gamma}^{(+)}(E,\varrho)=\frac{\hslash^2}{2m} \lim_{r\to\varrho^{-}}\sum_{i=-\infty}^{\infty} \frac{P_i^{(\gamma)}(\varrho)P_i^{(\gamma)}(r)}{E_i-E}\,.&&\label{rmatrix_el_dir1}
\end{eqnarray}
Comparing equation (\ref{rmatrix2_dir}) to equation (\ref{rmatrix_dir1}), we
see that  (\ref{rmatrix_el_dir1}) defines the diagonal elements of the
$R$-matrix $\mathsf{R}_{\mathsf{b}}^{(+)}$. Exactly as in the case of the
nonrelativistic $R$-matrix (compare equation (\ref{rmatrix_el})), the
relativistic $R$-matrix is expressed by a functional series whose convergence
is directly related to the convergence properties of the series
(\ref{exp_dir}). In the relativistic case the same question arises: {\it Does
the interchange of the limit and the infinite sum in equation
(\ref{rmatrix_el_dir1}) still give the same result?} In other words, does the
following identity hold:

\begin{eqnarray}\label{rmatrix_quest}
\lim_{r\to\varrho^{-}}\sum_{i=-\infty}^{\infty}
\frac{P_i^{(\gamma)}(\varrho)P_i^{(\gamma)}(r)}{E_i-E}
&\stackrel{\Large{?}}{=}& \sum_{i=-\infty}^{\infty}
\frac{P_i^{(\gamma)}(\varrho)P_i^{(\gamma)}(\varrho)}{E_i-E}\,?
\end{eqnarray}
The expression on the right-hand side of the above equation is traditionally
called the $R$-matrix. However, in the next section we show that in the case of
relativistic scattering it is {\it not} allowed to exchange the two operations.
Therefore the ``$R$-matrix'' as defined on the right-hand side of equation
(\ref{rmatrix_quest}) cannot connect the upper and lower component as in
equation (\ref{rmatrix_dir1}).


\section{Convergence theorems}\label{convergence}
In the previous section, we have reviewed the $R$-matrix theories for the
Schr\"odinger and Dirac particles. In both cases the $R$-matrix has been
defined as a limit of a certain eigenfunction expansion. In the nonrelativistic
theory it is given by equation (\ref{rmatrix_el}), whereas in the relativistic
theory by equation (\ref{rmatrix_el_dir1}). The problem of convergence of these
two functional series is equivalent to convergence of expansions
(\ref{exp_schr}) and (\ref{exp_dir}), respectively. Let us then recall the
convergence theorems for these two problems starting with the one for the
Schr\"odinger problem \cite{Titchmarsh}.
%
%
\begin{thm}\label{SturmLiouvile_th}(adapted from Ref. \onlinecite{Titchmarsh})
Consider the set of solutions of the following Sturm--Liouville problem:
\begin{eqnarray}
&&\left[-\frac{\mathrm{d}^{2}}{\mathrm{d}x^{2}}+q(x)-\lambda_n\right]y_n(x)=0,\qquad a\leq x\leq b,\label{SturmLiouvile_eq}\\
&&\left\{
\begin{array}{l}
y_n(a) \cos\alpha + y_n'(a)\sin\alpha=0\\
y_n(b) \cos\beta  + y_n'(b)\sin\beta=0
\end{array}\right.,\label{SturmLiouvile_bound}
\end{eqnarray}
where $q$ is assumed to be a real continuous function, and $\alpha,\beta \in
[0,\pi)$. Suppose now that $f$ is a real, continuous function of bounded
variation in the interval $[a,b]$. Then the expansion of $f$ in the
eigenfunctions of the eigenproblem
(\ref{SturmLiouvile_eq})+(\ref{SturmLiouvile_bound}) reads:
\begin{equation}
\bar{f}(x)=\sum_{n=0}^{\infty}c_n y_n(x),
\end{equation}
with
\begin{equation}
c_n=\int_{a}^{b}\mathrm{d}x\, y_n(x) f(x),
\end{equation}
where the eigenfunctions $y_n$ are assumed to be normalized to unity in
$\mathcal{L}^2([a,b])$. The series $\bar{f}(x)$ converges uniformly to $f(x)$
in the open interval $(a,b)$ [i.e., on each closed interval contained in
$(a,b)$]. Moreover, for $x=a$ and $x=b$ the following relations are true:
\begin{equation}
\bar{f}(a)=\left\{
\begin{array}{lc}
  f(a), & \mathrm{if}\,\alpha\neq 0\\
  0, & \mathrm{if}\,\alpha=0
\end{array}
\right.,\qquad
\bar{f}(b)=\left\{
\begin{array}{lc}
  f(b), & \mathrm{if}\,\beta\neq 0\\
  0, & \mathrm{if}\,\beta=0.
\end{array}
\right.
\end{equation}

\end{thm}
\noindent The above theorem can be applied to the nonrelativistic expansion,
indicating that in equation (\ref{rmatrix_el}) the limit can be exchanged with
the infinite sum, giving
\begin{equation}\label{rmatrix_el_lim}
\mathsf{R}_{\mathsf{b}\gamma}(E,\varrho)=\frac{\hslash^2}{2m} \sum_{i=0}^{\infty} \frac{P_i^{(\gamma)}(\varrho)P_i^{(\gamma)}(\varrho)}{E_i-E}.
\end{equation}
%
%

In what follows, we present a theorem concerning the expansion of a two component function appearing in the Dirac $R$-matrix definition.
\begin{thm}\label{Dirac_th}(adapted from Ref. \onlinecite{BirkhoffLanger})
Consider the following boundary value problem:
\begin{eqnarray}
&&\frac{\mathrm{d}}{\mathrm{d}x}\left(
  \begin{array}{c}
    u(x) \\
    v(x)
  \end{array}
\right)=\Big[\lambda A(x)+B(x)\Big]\left(
  \begin{array}{c}
    u(x) \\
    v(x)
  \end{array}
\right),\qquad a\leq x\leq b,\label{bv1}\\
&&W_a
\left(
    \begin{array}{c}
    u(a) \\
    v(a)
  \end{array}
\right)+
W_b
\left(
    \begin{array}{c}
    u(b) \\
    v(b)
  \end{array}
\right)=0,\label{bv2}
\end{eqnarray}
where $A(x),B(x)$ are $2\times 2$ matrices of functions continuous with their
first derivatives, $A(x)$ being diagonal, and $W_a,W_b$ are constant $2\times
2$ square matrices. Let us assume that: (i) the eigenvalues of $A(x)$, denoted
by $\vartheta_j(x)$ $(j=1,2)$, are continuous functions fulfilling the
following conditions for all $x\in[a,b]$:
\begin{eqnarray}\label{vartheta_cond}
&&\vartheta_j(x)\neq 0\nonumber\\
&&\vartheta_1(x)\neq \vartheta_2(x)\nonumber\\
&&\mathrm{Arg} [\vartheta_1(x)-\vartheta_2(x)]=\mathrm{const}\nonumber\\
&&\mathrm{Arg}\, \vartheta_j(x)=\mathrm{const}
\end{eqnarray}
\begin{figure}
  a)\includegraphics[width=0.4\textwidth]{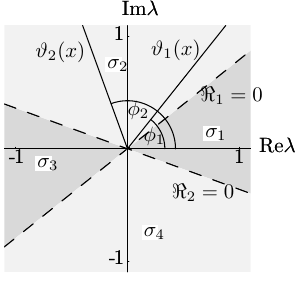}\,b)\includegraphics[width=0.4\textwidth]{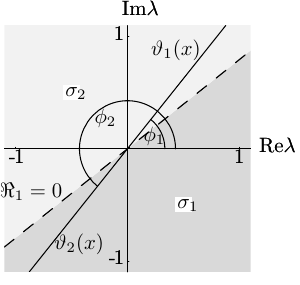}\\
  \caption{Two possible ways of dividing a complex plane of the parameter $\lambda$ into sectors in which the value $\Re_j$ has a fixed
  sign. In figure a) the two eigenvalues fulfill the condition
$\mathrm{Arg}[\vartheta_1(x)]\neq \mathrm{Arg}[\pm\vartheta_2(x)]$ while in b)
the two eigenvalues have phases differing by $\pi$, i.e.,
$\mathrm{Arg}[\vartheta_1(x)]=\mathrm{Arg}[-\vartheta_2(x)]$. The same division
is given by the situation
$\mathrm{Arg}[\vartheta_1(x)]=\mathrm{Arg}[\vartheta_2(x)]$, but it is not
depicted in the figure to avoid repetition. The continuous halflines mark the
rays on which the values of $\vartheta_j(x)$ lie, whereas the dashed lines
--- the rays on which the functions $\Re_j$ are equal to zero. The dashed lines
are the borders of the sectors $\sigma_k$, $k=1,2,\ldots$.}\label{podzial}
\end{figure}
Further, let us divide the complex plane of the parameter
$\lambda$ into sectors in which the sign of the expressions
$\Re_j=\mathrm{Re}(\lambda \int_{a}^{b}\mathrm{d}t\vartheta_j(t))$
$(j=1,2)$ is fixed. If one takes into account the conditions
(\ref{vartheta_cond}), there are two possibilities of dividing the
complex plane of $\lambda$ corresponding to situations when either
$\mathrm{Arg}[\vartheta_1(x)]\neq \mathrm{Arg}[\pm\vartheta_2(x)]$
or $\mathrm{Arg}[\vartheta_1(x)]=\mathrm{Arg}[\pm\vartheta_2(x)]$.
Both divisions are visualized in figure \ref{podzial}, where the
aforementioned sectors are denoted by $\sigma_k$, $k=1,2,\ldots$.
For each sector $\sigma_k$ we define the $2\times 2$ matrices
$\overline{\,\mathrm{I}}^{\,k}$ and
$\underline{\,\mathrm{I}}^{\,k}$ with elements
\begin{eqnarray}\label{deltas}
(\overline{\,\mathrm{I}}^{\,k})_{ij}=\left\{
\begin{array}{cc}
  \delta_{ij}, &  \Re_j\leq 0 \,\;\mathrm{in\; sector}\;\sigma_k \\
  0, &  \Re_j>0 \,\;\mathrm{in\; sector}\;\sigma_k
\end{array}
\right., &&\nonumber\\
(\underline{\,\mathrm{I}}^{\,k})_{ij}=\left\{
\begin{array}{cc}
  0, &  \Re_j\leq 0 \,\;\mathrm{in\; sector}\;\sigma_k \\
  \delta_{ij}, &  \Re_j>0 \,\;\mathrm{in\; sector}\;\sigma_k
\end{array}
\right. .
\end{eqnarray}
(ii) Let the matrices $W_a$ and $W_b$ be such that in each sector $\sigma_k$
the following matrix is invertible:
\begin{equation}\label{bvcond}
\Omega_k \equiv W_a \overline{\,\mathrm{I}}^{\,k}
+
W_b
\underline{\,\mathrm{I}}^{\,k}.
\end{equation}
Then: 1. the set of eigenvalues $\lambda_n$ and respective normalized
eigenfunctions $(u_{n},v_{n})^T$ of the problem (\ref{bv1})+(\ref{bv2}) is
infinitely countable; the orthonormality relation is the following:
\begin{equation}
\int_{a}^{b}\mathrm{d}x (\hat{u}_n(x),\hat{v}_n(x))A(x)
\left(\begin{array}{c}
u_m(x)\\
v_m(x)
\end{array}\right)=\delta_{mn},
\end{equation}
where $\big(\hat{u}_{n}(x),\hat{v}_n(x)\big)$ is the eigenvector of the adjoint
boundary value problem \cite{footnote}, corresponding to the eigenvalue
$\lambda_n$.
2. the development of any two-component function $F=(F_1,F_2)^{T}$, real and
continuous with the first derivative in the interval $[a,b]$, in the
eigenfunctions of the boundary problem (\ref{bv1})+(\ref{bv2}) is given by
\begin{equation}\label{exp_th}
\bar{F}(x)=\sum_{n=-\infty}^{\infty} C_n
\left(\begin{array}{c}
    u_{n}(x) \\
    v_{n}(x)
  \end{array}
\right),
\end{equation}
with
\begin{equation}\label{coeff_th}
C_n=\int_{a}^{b}\mathrm{d}x\, \big(\hat{u}_{n}(x),\hat{v}_n(x)\big) A(x)\left(\begin{array}{c}
    F_1(x) \\
    F_2(x)
  \end{array}
\right).
\end{equation}
The expansion (\ref{exp_th}) has the following properties:
\begin{subequations}
\begin{eqnarray}
\bar{F}(x)&=&F(x),\qquad \mathrm{for}\quad a<x<b,\label{limitx}\\
\bar{F}(a)&=&H_a F(a)+J_a F(b),\label{limita}\\
\bar{F}(b)&=&H_b F(a)+J_b F(b),\label{limitb}
\end{eqnarray}
\end{subequations}
where the $2\times 2$ matrices $H_a, J_a, H_b, J_b$ are fully determined by the matrix $A(x)$ and boundary conditions, and given by the expressions
\begin{subequations}
\begin{eqnarray}\label{matr_coeff}
H_a&=&\frac{1}{2}\mathbbm{1}_2+\sum_k \left[-\frac{\omega_k}{2\pi} \overline{\,\mathrm{I}}^{\,k} \Omega_k^{-1}W_a \underline{\,\mathrm{I}}^{\,k}\right]\label{Ha}\\
J_a&=&\sum_k \left[-\frac{\omega_k}{2\pi} \overline{\,\mathrm{I}}^{\,k} \Omega_k^{-1}W_b \overline{\,\mathrm{I}}^{\,k}\right]\label{Ja}\\
H_b&=&\sum_k \left[-\frac{\omega_k}{2\pi} \underline{\,\mathrm{I}}^{\,k} \Omega_k^{-1}W_a \underline{\,\mathrm{I}}^{\,k}\right]\label{Hb}\\
J_b&=&\frac{1}{2}\mathbbm{1}_2+\sum_k \left[-\frac{\omega_k}{2\pi} \underline{\,\mathrm{I}}^{\,k} \Omega_k^{-1}W_b \overline{\,\mathrm{I}}^{\,k}\right]\label{Jb}.
\end{eqnarray}
\end{subequations}
In the above, the parameter $\omega_k$ is an angle between the boundary rays of
a sector $\sigma_k$ (dashed lines in Fig. \ref{podzial}) and $\Omega_k$ is the
invertible matrix defined in equation (\ref{bvcond}).
\end{thm}
\noindent The above theorem is the special case of a more general eigenfunction
expansion problem considered in \onlinecite{BirkhoffLanger}. The reader may
find there the proof of the above facts. We will present here in more details
the situation directly applicable to the $R$-matrix expansion. Let us then
state and prove the following corollary.
\begin{cor}\label{Dirac_cor}
Consider the boundary value problem
\begin{eqnarray}
&&\left(\!
  \begin{array}{cc}
    p(x)-\lambda \rho(x) & -\mathrm{d}/\mathrm{d}x+t(x) \\
    \mathrm{d}/\mathrm{d}x+t(x) & q(x)-\lambda \rho(x)
  \end{array}\!\right)
\left(\!
  \begin{array}{c}
    f(x) \\
    g(x)
  \end{array}\!
\right)
=0, \quad a\leq x\leq b,\label{cor_eq}\\
%
&&\left(\!\begin{array}{cc} \cos \alpha  & \sin \alpha  \\
 0 & 0
\end{array}\right)
\left(
    \begin{array}{c}
    f(a) \\
    g(a)
  \end{array}\!
\right)+
\left(\!\begin{array}{cc}
 0 & 0 \\
 \cos \beta  & \sin \beta  \\
\end{array}\!
\right)
\left(\!
    \begin{array}{c}
    f(b) \\
    g(b)
  \end{array}\!
\right)=0,\label{cor_bv}
\end{eqnarray}
with $p,q,t,\rho$ being real functions continuous with first derivatives,
$\rho(x)>0$ for all $x\in[a,b]$ and $\alpha,\beta\in[0,\pi)$. Then the set of
eigenvalues $\lambda_n$ and eigenfunctions $(f_n,g_n)^T$ is discrete. The
expansion of a two-component function $F=(F_1,F_2)^{T}$, continuous with the
first derivative in $[a,b]$, in the set $\{(f_n,g_n)^T\}$ is given by
\begin{eqnarray}\label{cor_exp}
&&\bar{F}(x)=\sum_{n=-\infty}^{\infty} C_n
\left(\!\begin{array}{c}
    f_{n}(x) \\
    g_{n}(x)
  \end{array}\!
\right),\nonumber\\
&&C_n=\int_{a}^{b}\mathrm{d}x\, \rho(x) \big(f_{n}(x),g_n(x)\big)
\left(\!\begin{array}{c}
    F_1(x) \\
    F_2(x)
  \end{array}\!
\right),
\end{eqnarray}
and has the following properties:
\begin{subequations}
\begin{eqnarray}
\bar{F}(x)&=&F(x),\qquad \mathrm{for}\quad a<x<b,\label{cor_lim1}\\
\bar{F}(a)&=& \frac{1}{2}\left(\!
\begin{array}{ll}
 1-\cos 2 \alpha  & -\sin 2 \alpha  \\
 -\sin 2 \alpha  & 1+\cos 2 \alpha
\end{array}\!
\right) F(a),\label{cor_lim2}\\
\bar{F}(b)&=& \frac{1}{2}\left(\!
\begin{array}{ll}
 1-\cos 2 \beta  & -\sin 2 \beta  \\
 -\sin 2 \beta  & 1+\cos 2 \beta
\end{array}\!
\right)F(b).\label{cor_lim3}
\end{eqnarray}
\end{subequations}

\end{cor}
\begin{proof}
To prove the corollary, let us rewrite equation (\ref{cor_eq}) in such form
that the results from \onlinecite{BirkhoffLanger}, recalled in Theorem
\ref{Dirac_th}, apply directly, i.e. we would like to have the differential
equation and boundary conditions in the form (\ref{bv1})+(\ref{bv2}).
To achieve this, we multiply equation (\ref{cor_eq}) on the left-hand side by the unitary matrices $U$ and $\widetilde{U}$ given by
\begin{equation}
U=\frac{1}{\sqrt{2}}\left(\!\begin{array}{cc}
  \mathrm{i} & 1 \\
  1 & \mathrm{i}
\end{array}\!\right),\qquad \widetilde{U}=\left(\!\begin{array}{cc}
  \mathrm{i} & 0 \\
  0 & -\mathrm{i}
\end{array}\!\right),
\end{equation}
obtaining
\begin{equation}\label{eq_B}
\frac{\mathrm{d}}{\mathrm{d}x}
\left(\!
\begin{array}{c}
   u(x) \\
   v(x)
\end{array}\!\right)=
\left[\lambda
\left(\!\begin{array}{cc}
  \mathrm{i} \rho(x) & 0 \\
  0  & -\mathrm{i} \rho(x)
\end{array}\!\right)+
B(x)\right]
\left(\!
\begin{array}{c}
   u(x) \\
   v(x)
\end{array}\!\right),
\end{equation}
where $u=(\mathrm{i} f+g)/\sqrt{2}$ and $v=(f+\mathrm{i} g)/\sqrt{2}$, the diagonal matrix on the right-hand side corresponds to the matrix $A$, and the matrix $B$ contains the functions $p,q,t$, and fulfills the assumptions of Theorem \ref{Dirac_th}. At the same time, we have to adjust the boundary conditions to the functions $u$ and $v$. These become
\begin{equation}\label{bc_B}
\left(\!
\begin{array}{cc}
  \mathrm{e}^{\ii\alpha}& \ii \mathrm{e}^{-\ii\alpha} \\
  -\ii \mathrm{e}^{\ii\alpha}& \mathrm{e}^{-\ii\alpha}
\end{array}\!\right)
\left(\!
\begin{array}{c}
  u(a) \\
  v(a)
\end{array}\!\right)+
\left(\!
\begin{array}{cc}
  -\ii \mathrm{e}^{\ii\beta}& \mathrm{e}^{-\ii\beta} \\
  \mathrm{e}^{\ii\beta}& \ii \mathrm{e}^{-\ii\beta}
\end{array}\!\right)
\left(\!
\begin{array}{c}
  u(b) \\
  v(b)
\end{array}\!\right)
=0.
\end{equation}
Comparing equation (\ref{bc_B}) to (\ref{bv2}), we can see that the $2\times 2$ matrices on the left-hand side can be identified with $W_a$ and $W_b$, respectively. Let us then check if equation (\ref{eq_B}) with boundary conditions (\ref{bc_B}) satisfy the assumptions (i) and (ii) of Theorem \ref{Dirac_th}.

First, we will find the division of the complex plane of the
parameter $\lambda$ into sectors, as described in Theorem
\ref{Dirac_th}. The matrix $A(x)$ has two complex eigenvalues
$\vartheta_1(x)=\ii\rho(x)$ and $\vartheta_2(x)=-\ii\rho(x)$. Note
that since $\rho(x)$ is strictly greater than zero and continuous
for all $x\in[a,b]$, they satisfy the condition (i) of Theorem
\ref{Dirac_th}. Note, in particular, that
$\mathrm{Arg}[\vartheta_1(x)]=\mathrm{Arg}[-\vartheta_2(x)]$,
which leads to the division of the complex plane of the type shown
in Figure \ref{podzial}b, i.e. we have the following two sectors:
\begin{equation}
\sigma_1=\{\lambda: \Re_1< 0 \wedge \Re_2 > 0 \}=\{\lambda: \mathrm{Im}\,\lambda \geq 0 \}
\end{equation}
and
\begin{equation}
\sigma_2=\{\lambda: \Re_1>0 \wedge \Re_2<0 \}=\{\lambda:\mathrm{Im}\,\lambda <0 \}.
\end{equation}
This implies that the matrices defined in (\ref{deltas}) are
\begin{eqnarray}
\overline{\,\mathrm{I}}^{\,1}=\left(\!
                     \begin{array}{cc}
                       1 & 0 \\
                       0 & 0 \\
                     \end{array}\!
                   \right), &\qquad &
\underline{\,\mathrm{I}}^{\,1}=\left(\!
                     \begin{array}{cc}
                       0 & 0 \\
                       0 & 1 \\
                     \end{array}\!
                   \right),\\
\overline{\,\mathrm{I}}^{\,2}=\left(\!
                     \begin{array}{cc}
                       0 & 0 \\
                       0 & 1 \\
                     \end{array}\!
                   \right), &\qquad &
\underline{\,\mathrm{I}}^{\,2}=\left(\!
                     \begin{array}{cc}
                       1 & 0 \\
                       0 & 0 \\
                     \end{array}\!
                   \right).\nonumber
\end{eqnarray}
Now the condition (\ref{bvcond}) can be checked easily. Let us write out explicitly the matrices $\Omega_1$ and $\Omega_2$ for $W_a$ and $W_b$ as defined in (\ref{bc_B}), since we are going to use them in further calculations:
\begin{eqnarray}
\Omega_1=\left(\!
           \begin{array}{cc}
             \mathrm{e}^{\ii \alpha} & \mathrm{e}^{-\ii \beta} \\
             -\ii \mathrm{e}^{\ii\alpha} & \ii \mathrm{e}^{-\ii\beta} \\
           \end{array}\!
         \right),
&\qquad&
\Omega_2=\left(\!
           \begin{array}{cc}
             -\ii \mathrm{e}^{\ii \beta} & \ii \mathrm{e}^{-\ii \alpha} \\
             \mathrm{e}^{\ii\beta} & \mathrm{e}^{-\ii\alpha} \\
           \end{array}\!
         \right).
\end{eqnarray}
Clearly, they are invertible for all values of $\alpha, \beta\in\mathbbm{R}$, so the condition (ii) from Theorem \ref{Dirac_th} is fulfilled.

We have checked that all assumptions of Theorem \ref{Dirac_th} are satisfied,
therefore, the set $\{(u_n,v_n)^{T}\}$ of eigenfunctions of the problem
(\ref{eq_B})+(\ref{bc_B}) [and at the same time the set $\{(f_n,g_n)^{T}\}$ of
eigenfunctions of the problem (\ref{cor_eq})+(\ref{cor_bv})] is countably
infinite. Moreover, part 2 of the theorem applies to the expansion
(\ref{cor_exp}). To show explicitly that formulas
(\ref{cor_lim1})--(\ref{cor_lim3}) are valid, we premultiply equation
(\ref{cor_exp}) by $\sqrt{2}U$ obtaining the series of $(u_n,v_n)^{T}$
representing a complex two-component function $G=(G_1,G_2)^{T}=(F_2+\ii F_1,
F_1+\ii F_2)^{T}$:
\begin{equation}\label{cor_exp_G}
\bar{G}(x)=\sum_{n=-\infty}^{\infty} C_n
\left(\!\begin{array}{c}
    u_{n}(x) \\
    v_{n}(x)
  \end{array}\!
\right).
\end{equation}
Although part 2 of the theorem was formulated for real functions, the
generalization to complex functions is straightforward if the real and
imaginary parts are considered separately. Therefore, we will proceed with the
complex function $G$. Notice that coefficients $C_n$ can be written as
\begin{eqnarray}\label{cor_coeff_G}
C_n&=&\int_{a}^{b}\mathrm{d}x\, \rho(x) \big(f_{n}(x),g_n(x)\big)U^{\dagger}\widetilde{U}^{\dagger}\widetilde{U}U\left(\!\begin{array}{c}
    F_1(x) \\
    F_2(x)
  \end{array}\!\right)\nonumber\\
  &=&\int_{a}^{b}\mathrm{d}x\, \big(\hat{u}_{n}(x),\hat{v}_n(x)\big)A(x)\left(\!\begin{array}{c}
    G_1(x) \\
    G_2(x)
  \end{array}\!\right),
\end{eqnarray}
where $\big(\hat{u}_{n},\hat{v}_n\big)^T=(1/2)\widetilde{U}(u_n,v_n)^T$ is the
solution of the eigenproblem adjoint to (\ref{eq_B})+(\ref{bc_B}). The obtained
formula for coefficients is in agreement with equation (\ref{coeff_th}).
Consequently, the series (\ref{cor_exp}) modified with $U$ converges to
expressions (\ref{limitx}), (\ref{limita}), and (\ref{limitb}). Let us then
determine the matrices $H_a,J_a,H_b,J_b$ in this particular case. The angles
$\omega_1$ and $\omega_2$ are equal to $\pi$ since the sectors $\sigma_{1(2)}$
are half-planes (see Figure 1b), so exploiting equations (\ref{Ha})--(\ref{Jb})
we obtain

\vspace{5pt}
\begin{tabbing}
  aaaaaaaaaaaaaaaaaaa \= aaaaaaaaaaaaaaaaaaaaaaaaaaaaa \= aaaaaaaaaaaaaaaaaaaaaaa  \kill
  \> $H_a=\frac{1}{2}\left(
\begin{array}{cc}
 1 & -\ii \mathrm{e}^{-2 \ii \alpha} \\
 \ii \mathrm{e}^{2 \ii \alpha} & 1
\end{array}
\right),$ \> $H_b=\left(
\begin{array}{ll}
 0 & 0 \\
 0 & 0
\end{array}
\right),$ \\[2ex]
\> $J_a=\left(
\begin{array}{ll}
 0 & 0 \\
 0 & 0
\end{array}
\right),$ \> $J_b=\frac{1}{2}\left(
\begin{array}{cc}
 1 & - \ii \mathrm{e}^{-2 \ii \beta} \\
 \ii \mathrm{e}^{2 \ii \beta} & 1
\end{array}
\right).$
\end{tabbing}
\vspace{-64pt}
\begin{eqnarray}\label{matr_dir}
\,\nonumber\\
\,
\end{eqnarray}

\vspace{10pt}\noindent Taking into account the above results and equations
(\ref{limita}) and (\ref{limitb}), we obtain the following expressions for the
sum of series (\ref{cor_exp_G}) at the points $x=a$ and $x=b$:
\begin{equation}\label{sum_trans}
\bar{G}(a)=H_a G(a),\qquad \bar{G}(b)=J_b G(b).
\end{equation}
We recover the formulas for the original function $F$ premultiplying equations (\ref{sum_trans}) with $(1/\sqrt{2})U^{\dagger}$:
\begin{equation}
\bar{F}(a)=U^{\dagger}H_a U F(a),\qquad \bar{F}(b)=U^{\dagger}J_b U F(b).
\end{equation}
Finally, we insert the matrices (\ref{matr_dir}) into the above equations and obtain slightly modified formulas (\ref{cor_lim2}) and (\ref{cor_lim3}):
\begin{subequations}
\begin{eqnarray}
\bar{F}(a)&=&
\frac{1}{2}\left(\!
\begin{array}{cc}
 2\sin^2\alpha & -\sin 2\alpha \\
 -\sin 2\alpha & 2\cos ^2\alpha
\end{array}\!
\right)\left(\!
\begin{array}{c}
 F_1(a)\\
 F_2(a)
\end{array}\!
\right),\label{conva}\\
\bar{F}(b)&=&
\frac{1}{2}\left(\!
\begin{array}{cc}
 2\sin^2 \beta &-\sin 2\beta \\
 -\sin 2\beta & 2\cos^2 \beta
\end{array}\!
\right)\left(\!
\begin{array}{c}
 F_1(b)\\
 F_2(b)
\end{array}\!
\right)\label{convb}.
\end{eqnarray}
\end{subequations}
Summarizing, in the open interval $(a,b)$ the series (\ref{cor_exp}) converges to the function $F(x)$, whereas at the boundary points $x=a$ and $x=b$, the sum of the expansion is given by (\ref{conva}) and (\ref{convb}), respectively. This finishes the proof of Corollary \ref{Dirac_cor}.
\end{proof}

The Corollary reveals that the series (\ref{cor_exp}) converges to the function $F(x)$ in the whole interval $[a,b]$ if and only if the following two equalities hold simultaneously:
\begin{eqnarray}
\frac{1}{2}\left(\!
\begin{array}{cc}
 2\sin^2\alpha & -\sin 2\alpha \\
 -\sin 2\alpha & 2\cos ^2\alpha
\end{array}\!
\right)\left(\!
\begin{array}{c}
 F_1(a)\\
 F_2(a)
\end{array}\!
\right)=\left(\!
\begin{array}{c}
 F_1(a)\\
 F_2(a)
\end{array}\!
\right),\\
\frac{1}{2}\left(\!
\begin{array}{cc}
 2\sin^2 \beta &-\sin 2\beta \\
 -\sin 2\beta & 2\cos^2 \beta
\end{array}\!
\right)\left(\!
\begin{array}{c}
 F_1(b)\\
 F_2(b)
\end{array}\!
\right)=\left(\!
\begin{array}{c}
 F_1(b)\\
 F_2(b)
\end{array}\!
\right),
\end{eqnarray}
which is equivalent to the condition
\begin{equation}
F_1(a) \cos \alpha +F_2(a)\sin\alpha =0,\qquad F_1(b)\cos \beta +F_2(b)\sin\beta =0.
\end{equation}
One recognizes in this formulas the boundary conditions (\ref{cor_bv}).
Consequently, the sum of the series is continuous if and only if a function to
be expanded fulfills the same boundary conditions as the basis functions.

We demonstrated in this section that the properties of the developments into
eigensolutions of first-order differential systems (Theorem \ref{Dirac_th}) and
the expansions in the eigenfunctions of the Sturm--Liouville problem are
dramatically different. This fact, not realized by the originators of the
relativistic $R$-matrix method, has far-reaching consequences for the theory.
Although the assumptions in Corollary \ref{Dirac_cor} are restrictive, i.e. the
functions $p,q,t,\rho$, being elements of the matrices appearing in the
eigenproblem, are assumed to be continuous with their first derivative, its
conclusion still can be applied to expansion (\ref{exp_dir}) at the point
$r=\varrho$. This can be done, because the boundary conditions are separated
and, consequently, the convergence of the eigenfunction expansion at the point
$r=\varrho$ is independent of the behaviour at the point $r=0$. In fact, one
can consider the boundary-value problem (\ref{dir_eig})+(\ref{dir_eig_bc}) on
the interval $[\varepsilon,\varrho]$ with the boundary condition
$P_{i}^{(\gamma)}(\varepsilon)=0$ and show that the solution at the point
$r=\varrho$ remains the same for an arbitrarily small $\varepsilon$.
Taking this into account, we immediately see that the expansion (\ref{exp_dir})
for $r=\varrho$ does not, in general, converge to the solution of
(\ref{dir_rad}). It does only if the functions $P_{\gamma}(E,r)$ and
$Q_{\gamma}(E,r)$ satisfy the second of the boundary conditions
(\ref{dir_eig_bc}). This, however, cannot be assumed since this particular
condition does not have any physical meaning and is chosen in this way only to
obtain the expansion of the $R$-matrix. As a result, the commonly used
definition of the $R$-matrix contains an error, since in general it holds that
\begin{eqnarray}\label{rmatrix_neq}
\lim_{r\to\varrho^{-}}\sum_{i=-\infty}^{\infty} \frac{P_i^{(\gamma)}(\varrho)P_i^{(\gamma)}(r)}{E_i-E}
&\boldsymbol{\neq}&
\sum_{i=-\infty}^{\infty} \frac{P_i^{(\gamma)}(\varrho)P_i^{(\gamma)}(\varrho)}{E_i-E}.
\end{eqnarray}
The way to correct this mistake was found by Szmytkowski and Hinze for a
general multichannel case
\cite{SzmytkowskiHinze1995,SzmytkowskiHinze1996,Szmytkowski_rev,Szmytkowski1998}.
They introduced the correction which should by subtracted from the common and
faulty expression for the $R$-matrix in order to obtain the correct one.

\section{Conclusion}\label{conclusion}
Summarizing, we have provided theorems concerning the convergence of
eigenfunction expansions of a two-component function into eigenfunctions of a
Dirac operator on a finite closed interval augmented by separated boundary
conditions. In particular, we have shown that such expansions have
discontinuities at the boundary if the expanded function does not fulfill the
same boundary conditions as the basis functions. This confirms the result of
Szmytkowski \cite{JMathPhys} and has far-reaching consequences for the
relativistic $R$-matrix method. Moreover, the fact that the functional series
does not, in general, converge to a continuous function in the closed interval
may affect the rate of convergence and cause the Gibbs-like phenomenon
\cite{Barkhudaryan1,Barkhudaryan2} to occur.

The issue left as an open problem is the proof of convergence of
(\ref{cor_exp}) in the case when the functions $p(x),q(x),t(x)$ have a
singularity at one of the boundary points, e.g. at $x=a$. This is directly
related to convergence of expansion (\ref{exp_dir}) since the functions have
discontinuities for $r=0$. However, due to the separated character of the
boundary conditions (\ref{cor_bv}) and the fact that both the expanded function
and the basis functions vanish at $r=0$, the conclusions (\ref{cor_lim1}) and
(\ref{cor_lim3}) of Corollary \ref{Dirac_cor}, which concern the point
$r=\varrho$, should hold in this case, as well.


\begin{acknowledgements}
The author is indebted to Professor R. Szmytkowski for suggesting the problem,
valuable discussions, and commenting on the manuscript. I also thank R.
Augusiak for helpful discussions and comments. The research was supported by
the ``Universitat Aut\`{o}noma de Barcelona'' and Ministerio de Educaci\'{o}n
under FPU AP2008-03043.
\end{acknowledgements}

\newpage


\begin{thebibliography}{00}
\bibitem{footnote} The boundary value problem adjoint to (\ref{bv1})+(\ref{bv2}) has the form $$\frac{\mathrm{d}}{\mathrm{d}x}\left(
  \begin{array}{c}
    \hat{u}(x) \\
    \hat{v}(x)
  \end{array}
\right)=-\Big[\lambda A(x)+B(x)\Big]^{T}\left(
  \begin{array}{c}
    \hat{u}(x) \\
    \hat{v}(x)
  \end{array}
\right),\qquad a\leq x\leq b,$$
$$\widehat{W}_a
\left(
    \begin{array}{c}
    \hat{u}(a) \\
    \hat{v}(a)
  \end{array}
\right)+
\widehat{W}_b
\left(
    \begin{array}{c}
    \hat{u}(b) \\
    \hat{v}(b)
  \end{array}
\right)=0,$$
where the matrices appearing in the boundary conditions must satisfy $$W_a\widehat{W}_a^{T}=W_b\widehat{W}_b^{T}.$$
\end{thebibliography}
\end{document}